\documentclass[letterpaper, 10 pt, conference]{ieeeconf}

\IEEEoverridecommandlockouts
  \usepackage{amssymb}
  \usepackage{amsmath}
  \usepackage{multirow}
  \usepackage{graphics}
  \usepackage{pgf}
  \usepackage{cite}
\usepackage{times}
\usepackage{mathptmx}
\usepackage{cases}
\usepackage{ifthen}
\newboolean{ext}
\setboolean{ext}{false}
\newtheorem{Theorem}{Theorem}[section]

\newtheorem{Lemma}[Theorem]{Lemma}

\DeclareMathOperator{\sgn}{sgn}

\newcommand{\hide}[1]{}

\def\Sec {\S}
\def\Pr {\mathbb{P}}
\newcommand{\rem}[1]{}

\usepackage{setspace}
\title{\LARGE \bf An Analytical Approach to the Adoption of Asymmetric Bidirectional Firewalls: Need for Regulation?}
\author{MHR. Khouzani, Soumya Sen, Ness B. Shroff
\thanks{{\scriptsize MHR. Khouzani and Ness B. Shroff are with ECE department of the Ohio State University, Columbus, OH. email: {\tt\small khouzani,shroff@ece.osu.edu}. Soumya Sen is with EE department of Princeton University, Princeton, NJ. email: {\tt\small soumyas@princeton.edu}.}}%
}
\begin{document}

\maketitle
\thispagestyle{empty}
\pagestyle{empty}

\begin{abstract}
Recent incidents of cybersecurity violations have revealed the importance of having firewalls and other intrusion detection systems to monitor traffic entering and leaving access networks.  But the adoption of such security measures is often stymied by `free-riding' effects and `shortsightedness' among Internet service providers (ISPs). In this work, we develop an analytical framework that not only accounts for these issues but also incorporates technological factors, like asymmetries in the performance of bidirectional firewalls. Results on the equilibrium adoption and stability are presented, along with detailed analysis on several policy issues related to social welfare, price of anarchy, and price of shortsightedness.  
\end{abstract}

 \section{Introduction}

\ifthenelse{\boolean{ext}}{

For an individual in a network of interacting agents, the perceived utility of a
service can depend on the usage status of that service by others.
Such services are called to exhibit \emph{network externalities.}\footnote{In
general, externality refers to the side effects of a transaction
on individuals who has not been directly involved in the transaction. Such
effects can manifest as costs or benefits to stakeholders other than the
decision maker. Decision to increase one's safety is an example of positive
externality, i.e. where the side effects are perceived as benefit to others, and
polluting one's habitat is an example of negative externality, where others bear
the negative consequences as well.}
More often than not, the quality of such services improves as more individuals obtain it.
This effect is known as \emph{positive externality}. Positive externalities may
be restricted to those who have already adopted a service, e.g. Fax machines.
However, such positive externalities may exist even for those who have not obtained the
service. Security measures are among prominent examples of the latter.

A practically important problem for researchers has been to investigate the adoption
process of a new service with positive network externalities by taking into account the
incentives of the individuals.
In our progressively interconnected world, new technologies and applications are
constantly developed that are intertwined with networking capabilities; some
could only be born thanks to the development of the new communication infrastructure.
A next-generation network protocol, a social networking or a photo-sharing website are
among services with evident (positive) network externality. While they are common in positive externality phenomenon,
each of such services have their own specifics. For instance, a photo-sharing website has stand-alone value as well, while a new communication protocol does not.

\hide{Progressively, our societies take the form of collections of interconnected organisms.
Numerous services are intertwined with networking capabilities; some could only be born based on the expansion of the communication infrastructure.}

To investigate the adoption pattern of a new service, one needs to model the mechanism and the rationale behind the adoption decisions made by individuals. When a new technology/application is introduced, based on its cost and mostly stand-alone utility or under the influence of marketing attempts, it may be adopted by some individuals in the network.
It seems particularly unrealistic to assume that all individuals in a network are simultaneously faced with a decision of adoption, as in a single shot game setting.
Even the canonical repeated game model in which players repeat the single stage game (possibly for infinite number of times), keeping the account of the past decisions of the users and thinking strategically about the effect of their decision on the future decisions of the others, does not appear to capture what happens in reality either.
Closer to practice is to assume that individuals update their decisions at random times
based on their new estimate of the statistics of adoption.

The decision of an individual is determined by 
 comparison between the perceived utility under adoption versus refusal.
As more individuals adopt the new technology/application with positive externality, subsequent users measure a higher utility from the service. Adoption of the service by previous individuals may increase its derived utility for both adopters and non-adopters. If the positive externality for non-adopters grow faster with the adoption level of the technology than for adopters, the adoption level may converge to a fraction of the population.

Consider another case, where the benefits of adopting a new technology/application at its onset does not induce anyone to pay its price.
Adoption stays at zero. However, if adoption introduces a higher positive externality for adopters than for non-adopters, there might be an intermediate
level of adoption higher than which, the perceived quality of the service outweighs its cost. Consequently, seeding the market with this tipping point value (critical mass) will drive the adoption level to full market share.
As these intuitive arguments suggest, the outcome of adoption can be non-trivial. As we will discuss, more  cases are possible too.
}{
The growing incidents of information security breaches and attacks, presumably by both state (e.g., Stuxnet) and non-state actors (e.g., Anonymous) on various Government, ISP (Internet Service Providers), and e-commerce networks have heightened the need for such institutions to adopt effective security measures.  Cybersecurity has already been identified as \emph{``one of the most serious economic and national security threats"} by the US Government in creating its Comprehensive National Cybersecurity Initiative (CNCI) \cite{CNCI}. Nevertheless, the practical problem of security adoption faces several technological (e.g., firewall performance requirements), economic (e.g., costs of protection), and policy (e.g., regulating social welfare) challenges. 

Our work addresses these issues in the context of adoption of \emph{asymmetric, bidirectional} firewalls 
by autonomous systems (e.g., ISPs) in the Internet. We calculate the equilibrium adoption level(s) and their stability in terms of the measurable specifications of a firewall, and analyze how factors like shortsightedness in the adoption decisions of ISPs impact the overall security of such systems of interconnected networks.  

Our work provides insights for both policy guidelines and marketing of future firewalls, some of which are listed below:
\begin{list}{\labelitemi}{\leftmargin=0.5em}
\item In general, the equilibrium levels of adoption is not unique and depend on the initial seeding of the firewalls. 
However, for a given initial seeding, the equilibrium is unique and stable.
\item In the case of asymmetric bidirectional firewalls, performance improvements in identifying and blocking outgoing threats \emph{decreases} the overall firewall adoption level in the system as a consequence of the free-riding phenomenon.
\item\hide{We demonstrate that }The (centralized) \emph{social} optimum solution for firewall adoption results in  \emph{each} ISP having an \emph{individually} higher utility compared to the decentralized equilibrium.
\item\hide{ We show that }The optimum regulation on the amount of protection on the outgoing traffic in the \emph{decentralized} scenario is (rather counter-intuitively) providing \emph{no protection} at all. 
This is while the optimum protection in the \emph{centralized} scenario is providing protection on the outgoing traffic as high as the protection level on the incoming traffic. This result in turn suggests a need for regulation of \emph{both} firewalls and ISPs. 
\item\hide{  We show that }It is possible that the shortsightedness of the ISPs can in fact \emph{help} to improve the overall security of the network. We identify specific conditions under which shortsightedness in the
decision-making of ISPs can helps/hurts the overall system security, by introducing the new (general) notion of \emph{price of shortsightedness}.
\end{list}
}

 \paragraph*{Related Literature}
Research on network security adoption is closely related to three related topics: role of network externality, game-theoretic models for security, and epidemic diffusion. 
%
Previous works on the adoption of network technologies \cite{sen1,sen2} have studied the role of \emph{network externality} on the equilibrium outcome and diffusion dynamics. 
Epidemiologists have studied free-riding's impact on vaccination strategies
as some parents myopically choose not to vaccinate their children if the relative risks of contacting a disease given its current prevalence is less than the potential for damage from the vaccine \cite{d2007vaccinating}.
Similar results were also reported using game-theoretic tools  for modeling \emph{vaccination games} \cite{bauch2004vaccination,heal2005vaccination,reluga2011general}.
%
 %
In the context of network security adoption, 
game-theory has been used to analyze investments in network security \cite{grossklags2008secure, johnson2010uncertainty}. 
Others have proposed cyber-insurance as a complementary security mechanism \cite{bohme2006models}.  

Our work extends these contributions by (a) focusing on the question of ISP-wide adoption of security measures as opposed to individual user-level patching or cyberinsurance provisioning, (b)  modeling and analyzing the performance of firewalls with \emph{asymmetric, bidirectional} traffic monitoring capabilities. 
This distinction is particularly relevant from a policy perspective of whether a minimum performance is required for bidirectional firewalls, especially in view of the US Government's CNCI \cite{CNCI} policy of 
conducting \emph{``real-time inspection and threat-based decision-making on network traffic entering or leaving executive branch networks,"} and (c)~explicitly deriving regulation constraints for managing the free-riding and shortsightedness of the ISPs.

\section{System Model}\label{Sec:model}
An ISP provides a gateway that connects a subnet to the Internet.
A firewall software installed by an ISP can monitor the incoming and outgoing traffic for
blocking security breaches. 
%
Adoption of such security measures has a stand-alone benefit for an ISP. Moreover, it can slow down the rate of attacks and provides positive externality to the rest of the ISPs (adopters and non-adopters) by improving the overall security in the network. Namely, the nodes in other ISPs will be less likely to be targeted by an attack
originating from the subnet of the protected ISPs.
However, adoption of firewalls is not without cost: there can be a one purchase fee, as well as recurrent usage costs such as routine maintenance and updating, slower connection as a result of latencies introduced by traffic monitoring, \emph{false positives} and hence the occasional blocking of the legitimate traffic, etc.
\hide{maintained and routinely updated with new signatures. They also may render the connection
sluggish as the traffic is probed for anomalies that introduces more latencies.
In addition, firewalls have a false positive rate, that is, with
some small yet nonzero rate, they fail to distinguish between legitimate traffic and an intrusion attempt and erroneously block the legitimate traffic.}
In what follows we provide a practical model that captures the key attributes for adoption dynamics of firewalls.
Note that our goal is analyzing qualitative patterns and behavior of adoption. Hence, we make some technical assumptions along the way to keep the model analytically tractable.

We consider a continuous-time model with $N$ inter-connected ISP networks.
Once an ISP purchases the firewall, it can un-adopt it by simply disabling the firewall.
Subsequent adoptions are performed by enabling the firewall, and in particular, do
not entail paying the firewall's one time purchase fee. That is, the cost of adoption only
includes the recurrent usage cost of the firewall for subsequent adoptions of an ISP that
has purchased it. Hence, we need a model that distinguishes between the first
adoption and subsequent re-adoptions.
To do this, we stratify the ISPs into three types:
(1)~ISPs that have \emph{purchased} and \emph{enabled} the firewall;
(2)~ISPs that have \emph{not purchased} it; and
(3)~ISPs that have \emph{purchased} the firewall but have \emph{disabled} it.
Let the \emph{fraction} of ISPs of each of the above types at time $t$ be $x(t)$, $y(t)$ and $1-x(t)-y(t)$, respectively.
The pair $(y(t),x(t))$ represents the adoption state of the network at time $t$.

Each ISP independently updates its decision regarding the adoption of the firewall at independent random epochs that occur according to a Poisson processes with rate $\gamma$.  These are the epochs at which an ISP updates
its measure of the adoption state of the network and accordingly re-evaluates its
contingent utilities.  Specifically, the decisions of the ISPs are assumed to be their best response 
to the \emph{current} adoption state, that is, assuming the current level is not going to change.
Let $v_{ij}(x):[0,1]\to[0,1]$ be the probability with
which a decision-making ISP switches from state $i$ to $j$,
where $i,j\in\{n,e,d\}$ represent the state of the ISP with respect to adoption. 
We set the convention that $n$ indicates 
\underline{n}ot purchased, $e$ represents purchased and \underline{e}nabled, and $d$ denotes purchased but \underline{d}isabled.
Note that $v_{ij}$ is a function of $x(t)$ only (and not $y(t)$), since the ISPs that have not
obtained the firewall and those that have disabled it are functionally similar regarding the
protection against threats. 
For large $N$, the dynamics of $(y(t),x(t))$ is close to the solution of the following ODE:
\begin{numcases}
{}\, \dot{y}(t)=-\gamma y(t) v_{ne}(t)\label{sys_dyn_gen}\\
 \dot{x}(t)= \gamma y(t) v_{ne}(t) + \gamma (1-x(t)-y(t))v_{de}(t)- \gamma x(t) v_{ed}(t)\notag
\end{numcases}
The decision of the ISPs is determined by comparing the expected utilities given each decision.
Accordingly, define $G_{ij}(x)$ for $i,j\in\{n,e,d\}$ to be  the expected utility of an ISP if it  changes
its adoption state from $i$ to $j$  given the current fraction of ISPs with enabled firewalls is $x$. Simply put, $G_{ne}(x)$ is the expected
utility of the ISP at its decision-making epoch if it purchases and enables the firewall, $G_{nn}(x)$ is its
expected utility if it stays yet-to-purchase, and so on.
Let $c_0$ denote the one-time purchase fee of the firewall.
We note that $G_{ne}(x)$ and $G_{de}(x)$ differ only in the purchase fee of the firewall.
Specifically, $G_{de}(x)=G_{ne}(x)+c_0$.
Therefore, for $c_0>0$, we have $G_{de}(x)>G_{ne}(x)$.
Moreover, $G_{nn}(x)=G_{en}(x)$, as both utilities are only associated with the expected cost of the intrusions, which is the same for an ISP with its firewall disabled, and an ISP that is yet to install it.
Similarly, $G_{dd}=G_{ed}=G_{nn}$, and hence, without loss of generality and for brevity, we can define
 $G_{N}:=G_{nn}=G_{en}=G_{dd}=G_{ed}$.
Along the same line of arguments, 
we define $G_{E'}:=G_{ee}=G_{de}$ and $G_{E}:=G_{ne}$.

The introduction of $G_N$, $G_{E'}$ and $G_E$ leads to the following decision rules in~\eqref{sys_dyn_gen}:
If $G_{E'}(x(t))\neq G_{E}(x(t))$, then $v_{ed}=\mathbf{1}_{G_{E'}(x(t))>G_{E}(x(t))}$ and 
$v_{de}(x(t))=1-v_{ed}(x(t))$. If $G_{E'}(x(t))=G_{E}(x(t))$, then the ISP is indifferent between enabling and disabling the firewall and hence $v_{ed},v_{de}\in [0,1]$.
Similar comments apply to $v_{ne}$ and $v_{en}$.
A list of our main notations is provided in
Table~\ref{Table:Notations}.

\begin{table}
\centering
\caption{Main notations in the model}\label{Table:Notations}
 \begin{tabular}{c p{0.7\linewidth}}\hline
parameter & definition \\
\hline
$N$ & {\footnotesize number of ISPs}\\
$x(t)$ & {\footnotesize fraction of the ISPs at time $t$ that have adopted and enabled the firewall}\\
$y(t)$ & {\footnotesize fraction of the ISPs at time $t$ that are yet to purchase}\\
$\gamma$ & {\footnotesize rate at which each ISP updates its adoption decision}\\
$G_{ij}$ & {\footnotesize expected utility of an ISP if it chooses adoption status $j$
given that its current adoption status is $i$}\\
$v_{ij}$ & {\footnotesize the probability with which a decision-making ISP switches from adoption status $i$ to $j$}\\
$\rho$ & {\footnotesize rate at which a new intrusion attempt is launched}\\
$n_I$ & {\footnotesize number of attackers in the network}\\
$\mu$ & {\footnotesize rate at which a successful intrusion to a
subnet is detected and blocked}\\
$c_0$ & {\footnotesize one time purchase fee of the firewall}\\
$c$ & {\footnotesize per unit time usage cost of the firewall}\\
$C_{1I}$ & {\footnotesize instantaneous cost upon a successful intrusion}\\
$C_{2I}$ & {\footnotesize cost (loss/damage) per unit time of intrusion}\\ 
\hline
 \end{tabular}
\end{table}

\subsection{Finding the Equilibrium Points}
In order to obtain the equilibrium solutions (fixed points) of the adoption process as described by~\eqref{sys_dyn_gen}, we need to derive the utility functions $G_N(x)$, $G_{E}(x)$, and $G_{E'}(x)$.
As mentioned before, $G_{E'}(x)=G_E(x)+c_0$. Hence, we focus on calculating $G_N(x)$ and $G_E(x)$.
Along the way of calculations, some auxiliary symbols are defined to replace lengthy expressions. A list of the most important ones is provided in Table~\ref{Table:expressions} for ease of access.

Let $c$ be the cost per unit time of using the firewall incurred by an adopter ISP due to maintenance, communication latencies, false positives, etc.
Note that $c$ differs from $c_0$, i.e. the single-time purchase fee to obtain the firewall.
For simplicity of exposition, we consider security breaches that do not
propagate in the network. For example, we will \emph{not} consider attacks involving self-replicating
malicious codes (known as \emph{worms}) in this article.
Hacking is a typical example of a non-replicating type of attack. We will refer to such
attacks by the umbrella term of \emph{intrusion} attempts.
When a host in a subnet of an ISP is compromised, the ISP incurs an instantaneous cost of
$C_{1I}$ and a per unit cost of $C_{2I}$ that persist as long as the host is
infiltrated.
The instantaneous cost may reflect the losses due to exposure of private information\hide{such as credentials (fingerprints, voice recognition, passwords, etc), credit card information} or manipulation of data, while the per unit time cost can represent the accumulation of
eavesdropped data,\hide{such as keystroke logs,} accessing the network at the cost of the victim,
slowdown of the victim's machine or the ISP's service, etc.
The time it takes to remove the infection is according to an exponential random variable with rate
$\mu$. Machines are again susceptible to future attacks, since new attacks are likely to exploit new techniques.

Without prior knowledge, new security breaches 
can originate from the subnet of any of the ISPs. We assume that ISPs are
homogeneous, that is, they assign the same parameters for costs and have similar subnet
sizes, furthermore, that a target of an intrusion is chosen uniformly randomly from the space of IP addresses. These along with the assumption of homogeneous sizes of the subnets, imply that the target is equally likely to belong to the subnet of any of the ISPs.

The success probability of an intrusion attempt depends in part on the status of the ISPs of the
attacker as well as the ISP of the target with regard to the adoption of the firewall.\footnote{Note that intermediate routers do not monitor for threats and the only traffic monitoring for threats are at border (edge) ISPs.} Specifically, the
highest chance of intrusion success is when \emph{neither} of the ISPs have enabled the
firewall, while the lowest likelihood is when \emph{both} ISPs have (obtained and) enabled it.
Based on the four different conditions for the adoption status of the ISPs of an attacker and
its target, we define intrusion \emph{success probabilities} $\pi_0$,  $\pi_1$, $\Pi_0$
and $\Pi_1$ according to Table~\ref{Table:Omegas}. Namely, $\pi_1$ is the success probability of intrusion if both ISPs have enabled firewalls in place, $\pi_1$ is the success probability of an intrusion if only the target's ISP has adopted the firewall, and so forth.
\begin{table}
\centering
\caption{Success probabilities of an intrusion attempt}\label{Table:Omegas}
\begin{tabular}{cc|c|c|}
\cline{3-4}
& & \multicolumn{2}{|c|}{{\small Host's ISP}} \\ \cline{3-4}
& &  {\small Protected} & {\small Not Protected} \\ \cline{1-4}
\multicolumn{1}{|c|}{\multirow{2}{*}{{\small Attacker's ISP}}} &
\multicolumn{1}{|c|}{{\small Protected}} & $\pi_1$ & $\Pi_1$      \\ \cline{2-4}
\multicolumn{1}{|c|}{}                        &
\multicolumn{1}{|c|}{{\small Not Protected}} &$\pi_0$  & $\Pi_0$    \\ \cline{1-4}
\end{tabular}
\end{table}

Without loss of generality, we let $\Pi_0=1$ and only consider the attempts that
are successful in the absence of any firewall. 
However, we continue to use
the \emph{notation} $\Pi_0$ in our formulation for presentation purposes. 
In general, the following ordering holds for the intrusion success probabilities:
\begin{align*}
0\leq \pi_1 \leq \pi_0\leq \Pi_1\leq \Pi_0\leq1.
\end{align*}
That $\Pi_0$ is the largest of the group is obvious, as it is the probability of
success of an intrusion if no
firewall is set up on both ISPs of the attacker and target (hence the most exposed
scenario).
$\pi_0\leq \Pi_1$, since the primary goal of the firewall is to protect the subnet against the incoming threats and hence, 
a marketable firewall provides no less protection against the incoming threats than against the outgoing threats.
$\pi_1\leq \pi_0$, as $\pi_1$ is the success probability of an intrusion that has to
bypass \emph{both} firewalls of its own subnet's ISP and of the ISP of the
victim's machine, while $\pi_0$ is the success probability of an intrusion that only has to bypass the firewall of the victim's ISP.

For a firewall whose mechanism of intrusion prevention is only signature-based, if both firewalls have access to the same signature database then $\pi_1=\pi_0$, that is, if an intrusion can successfully bypass one of the firewalls, it will be able to bypass the other one as well. We will refer to this case as the \emph{mutually inclusive} scenario. However, it could be that one of the firewalls is more up-to-date than the other, hence it is
likely that $\pi_1<\pi_0$. Also, anomaly detection mechanisms are in essence
probabilistic and they have a \emph{false negative} chance. The past traffic history of
the two ISPs differ, hence the blocking events of the two firewalls may not be exactly mutually inclusive.
In case the  intrusion prevention outcomes of the firewalls are mutually independent, for $\Pi_0=1$,
we have $\pi_1=\pi_0\Pi_1$. Hence, we have the following additional structural inequality:
\begin{equation}0\leq
\pi_0\Pi_1\leq\pi_1.\label{omega_ineq_omega0_1}\end{equation}
New mechanisms are proposed in which the firewalls in different ISPs ``co-operate'' to
improve their detection and blocking chances (e.g.~\cite{zhang2006cooperative}). In such cases, it is possible for
$\pi_1$ to be less than $\pi_0\Pi_1$. We, however, do not consider such cases in the current article.

In our model, we assume that new intrusion attempts are initiated at a rate $\rho$ after independent and
exponentially distributed time intervals. \hide{Consistent with the assumption of $\Pi_0=1$,
these are the intrusion attempts that are \emph{successful} in the \emph{absence} of any
firewall.} Let $n_I$ be the number of active attackers in the network, which we assume to
be fixed. Hence, $n_I\rho$ is the rate of all incidents of successful intrusion attempts when no firewall is present. Note that although estimating $n_I$ and $\rho$ is
difficult,  $n_I\rho/N$ can be readily estimated by each ISP, since it is the rate of
successful intrusion attempts on a single ISP when the firewall is absent, and as we will see,
that is indeed what appears in the expression for the utilities of each ISP.
The rest of the parameters of the problem, such as $c_0$, $c$, $C_{1I}$, $C_{2I}$, etc., are also assumed to be publicly known.

The utility of an ISP is a decreasing function of costs and losses due to potential future
intrusions to its subnet.
For ease of calculations, we assume risk-neutral ISPs. Hence, we can directly take the
negative of the costs to be the ISPs' utility. 
Let $\sigma(t)$
represent the state of the decision-taking ISP with respect to the intrusion, that is,
$\sigma(t)=1/0$ indicates that the ISP's subnet is \emph{intruded/not intruded} at time
$t$. Without loss of generality (achieved by shifting the time coordinate) we can let $t=0$.
Now, $G_N$ can be computed as follows:
\begin{gather}
G_N= \sum_{i\in\{0,1\}}\left(G_N\mid \sigma(0)=i\right)\Pr\{\sigma(0)=i\}\label{G_0}.\\
\text{Define:}\ \ \ \  A:=-\left(G_N\mid \sigma(0)=1\right), \ \ B:=-\left(G_N\mid \sigma(0)=0\right).\notag
\end{gather}
In words, $A$ ($B$) is the conditional expected cost of non-adoption given that a node in the subnet of the ISP is  currently intruded (not intruded).
At a decision making epoch, the ISP may not be aware of an ongoing
successful intrusion. As we stated before, a successful intrusion is detected, and then immediately blocked, after an exponentially random delay with rate $\mu$.
To introduce a measure  of the short-sightedness (or impatience) of the ISP, we use the \emph{discount factor} $r$:  higher values of $r$ imply a greater degree of short-sightedness. Hence, we have:
\begin{equation}\label{A}
\begin{split}
A=&C_{1I}+\int_{0}^{\infty}\! e^{-\mu t}\mu \,dt
\left(\int_{0}^{t}\!e^{-r\tau}C_{2I}\,d\tau+e^{-rt}B\right)\\
=&C_{1I}+\frac{C_{2I}}{\mu+r}+\frac{\mu B}{\mu+r}.
\end{split}
\end{equation}
\hide{Note that in introducing $B$ in the above equation, we used the memoryless property of exponential random variable: the statistics of the system observed from some time in the middle of the healthy (clean) period, is
the same as the statistics observed from the start of a healthy period.}
\hide{In calculation of
$B$, we ignore the (small) probability that the attack to a subnet is originating from a
member of the same subnet because it does not affect the adoption decision of a risk-neutral ISP since internal attacks are unaffected by the presence of the firewall.}
On the other hand:
\begin{align*}
\!\!\!B=&\int_{0}^{\infty}\!\!\!\Pr\{\sigma(\tau)=0,0\leq\tau\leq
t\}\frac{n_I\rho}{N}(\Pi_1x+\Pi_0(1-x))Ae^{-rt}\,dt\\
 =&\int_{0}^{\infty}\!\!{\textstyle\big[\!\!\sum_{i=0}^{\infty}e^{-n_I\rho t}(n_I\rho t)^i
\big(\frac{N-1}{N}+\frac{x(1-\Pi_1)}{N} \frac{(1-x)(1-\Pi_0)}{N}\big)^i/i!\Big]\times}\\
&\phantom{whazzzaaaaaaaaaa.......}{\textstyle \frac{n_I\rho}{N}(\Pi_1x+\Pi_0(1-x))Ae^{-rt}\,dt}.
\end{align*}
Let $\eta:=\dfrac{n_I\rho}{N}\left({\Pi_1x+\Pi_0(1-x)}\right)$. Then:
\begin{flalign}
 B=&\int_{0}^{\infty}\!e^{-\eta\nu} \eta e^{-r\nu} A\,d\nu=\frac{\eta}{r+\eta}A.\label{B}
\end{flalign}
From~\eqref{A} and~\eqref{B}, we get:
\begin{gather}
A=C_{1I}+\frac{C_{2I}}{\mu+r}
+\frac{\mu}{r+\mu}\cdot\frac{\eta}{r+\eta}A\notag\\
\Rightarrow
A=\frac{C_{1I}+\frac{C_{2I}}{\mu+r}}{1-\frac{\mu\eta}{(r+\mu)(r+\eta)}}=\frac{a(r+\eta)}{
r+\mu+\eta},\ B=\frac{a\eta}{r+\mu+\eta},\label{A_expl,B_expl}
\end{gather}
where $a:=\left(C_{1I}(r+\mu)+C_{2I}\right)/r$ (listed in Table~\ref{Table:expressions}). 
Now, since as we assumed the ISPs make their decision as a best response the \emph{current} level of adoption assuming it will not change,  the probability of being intruded or not intruded that they assign themselves comes from the stationary distribution associated with the current level of adoption. Hence, from renewal theory:
\begin{equation}\label{P_A}
\Pr\{\sigma(1)=1\}=\frac{\frac{n_I}{N}\rho(x\Pi_1+(1-x)\Pi_0)}{\mu+\frac{n_I}{N}
\rho(x\Pi_1+(1-x)\Pi_0)}=\frac{\eta}{\eta+\mu}
\end{equation}
\begin{equation}\label{P_B}
\Pr\{\sigma(0)=0\}=\frac{\mu}{\mu+\frac{n_I}{N}\rho(x\Pi_1+(1-x)\Pi_0)}=\frac{\mu}{
\eta+\mu}
\end{equation}
Take $\Lambda:=\dfrac{\rho n_I}{N}$, which we refer to as the \emph{intensity} of
intrusion attempts.
Combining~\eqref{G_0}, \eqref{A_expl,B_expl},  \eqref{P_A}, \eqref{P_B} yields:
\begin{align}
G_N(x)=&
-\frac{a(r+\eta)\eta}{(r+\mu+\eta)(\eta+\mu)} -
\frac{a\eta\mu}{(r+\mu+\eta)(\eta+\mu)}\notag\\=&-\frac{a\eta}{\mu+\eta}=-a\frac{\Lambda\left({
\Pi_1x+\Pi_0(1-x)}\right)}{\mu+\Lambda\left({\Pi_1x+\Pi_0(1-x)}\right)}\label{G0_explicit}
\end{align}

\begin{table}
 \centering
\caption{Definition of the main auxiliary coefficients}
\begin{tabular}{cc}\hline
coefficient & definition \\
\hline
$a$ & $\left(C_{1I}(r+\mu)+C_{2I}\right)/r$\\
$\eta$ & $\dfrac{n_I\rho}{N}\left({\Pi_1x+\Pi_0(1-x)}\right)$\\
$\Lambda$ & $\dfrac{n_I\rho}{N}$\\
\hline
\end{tabular}
\label{Table:expressions}
\end{table}

Calculation of the $G_E(x)$  is now straightforward: there are two components, first one is the
cost of adoption, that is simply $c_0+\dfrac{c}{r}$, and the other component is the cost of
intrusion, whose calculation follows the same steps as in $G_N(x)$ except for
accordingly changing the probabilities of successful intrusion. Hence, we can directly
deduce:
\begin{align}
G_E(x)=-c_0-\frac{c}{r}-a\frac{\Lambda\left({\pi_1x+\pi_0(1-x)}\right)}{
\mu+\Lambda\left({\pi_1x+\pi_0(1-x)}\right)}\label{G1_explicit}
\end{align}

A straightforward yet important property of the expected utilities is that they are (both) increasing in the value of $x$:
\begin{Lemma}\label{Lem:positive_externality}
 For any $x\in[0,1]$ we have: $\dfrac{\partial G_N(x)}{\partial x}$, $\dfrac{\partial G_E(x)}{\partial x}$,
$\dfrac{\partial G_{E'}(x)}{x}\geq 0$.
The equality holds only if $\Pi_1=\Pi_0$.
\end{Lemma}
Hence, positive externality exists for both adopters and non-adopters. 
Moreover, the positive externalities vanish only when there is no protection against the outgoing threats.

\subsection{Multiplicity/Uniqueness and Stability of the Equilibrium}\label{subsec:uniqueness}
We start this section with a lemma that we will use multiple times.
\begin{Lemma}\label{Lem:D(x^*)<0}
 For non-cooperating firewalls (hence inequality~\eqref{omega_ineq_omega0_1}) and when $\Pi_1<\Pi_0$, $D(x):=\dfrac{d}{\,dx}(G_E(x)-G_N(x))<0$ at any point $x\in[0,1]$.
\end{Lemma}

Note that this lemma also implies $\dfrac{d}{\,dx}(G_{E'}(x)-G_N(x))<0$ for any $x\in[0,1]$, since $G_{E'}(x)=G_E(x)+c_0$. Also, for
the case of $\Pi_1=\Pi_0$, we have $\pi_1=\pi_0$. This leads to
$D(x)=0$ for all $x$.
\hide{The proof of this lemma follows the following steps: first, $D(x)$ is expanded using~\eqref{G0_explicit} and~\eqref{G1_explicit} as:
\begin{gather}
D(x)=\dfrac{d}{\,dx}(G_E(x)-G_N(x)) = {\mu a \Lambda}\times\label{Dx_expanded}\\
\left[\frac{\left({\pi_0-\pi_1}\right)}{[\mu+\Lambda
\left({\pi_1x+\pi_0(1-x)}\right)]^2}-\frac{\left({\Pi_0-\Pi_1}\right)}{\left[
\mu+\Lambda\left({\Pi_1x+\Pi_0(1-x)}\right)\right]^2}\right]\notag
\end{gather}
Subsequently, the maximum value of $D(x)$ as a (multi-variable) function of $(\pi_0,\pi_1,\Pi_0,\Pi_1)$ subject to the structural constraint of $\Pi_1\pi_0\leq
\pi_1\leq \pi_0\leq \Pi_1$ is found and shown to be negative. }

\begin{proof}
From~\eqref{G0_explicit} and~\eqref{G1_explicit} we have:
\begin{gather}
D(x)=\dfrac{d}{\,dx}(G_E(x)-G_N(x)) = {\mu a \Lambda}\times\label{Dx_expanded}\\
\left[\frac{\left({\pi_0-\pi_1}\right)}{[\mu+\Lambda
\left({\pi_1x+\pi_0(1-x)}\right)]^2}-\frac{\left({\Pi_0-\Pi_1}\right)}{\left[
\mu+\Lambda\left({\Pi_1x+\Pi_0(1-x)}\right)\right]^2}\right]\notag
\end{gather}
%
Recall that without loss of generality, we can let $\Pi_0=1$ and hence, we have the
inequality in~\eqref{omega_ineq_omega0_1}, that is, $\pi_0\Pi_1\leq \pi_1\leq
\pi_0$. Let $\phi(\pi_0,\pi_1):=
\left[\dfrac{\left({\pi_0-\pi_1}\right)}{[\mu+\Lambda
\left({\pi_1x+\pi_0(1-x)}\right)]^2}\right].$
Now, consider the following maximization problem:
\begin{align*}
 \max_{\pi_1,\pi_0} \phi(\pi_0,\pi_1)\quad\text{s.t.}\ \ \Pi_1\pi_0\leq
\pi_1\leq \pi_0\leq \Pi_1
\end{align*}
For any fixed value of $\pi_0$, $\phi(\pi_0,\pi_1)$ is strictly decreasing in
$\pi_1$.
Also, for any fixed value of $\pi_1$, $\phi(\pi_0,\pi_1)$ is strictly increasing in
$\pi_0$.
Hence, the maximum of $\phi(\pi_0,\pi_1)$ is achieved at the extreme point of
$(\Pi_1,\Pi_1^2)$.
Therefore, from~\eqref{Dx_expanded}:
\begin{multline*}
D\leq (1-\Pi_1) \{\frac{\Pi_1}{[\mu+\Lambda\Pi_1
\left({\Pi_1x+(1-x)}\right)]^2}\\-\frac{1}{\left[\mu+\Lambda\left({\Pi_1x+(1-x)}
\right)\right]^2}\}
\end{multline*}
Now, consider the function $\!\psi(y)\!=\!\dfrac{y}{[\mu+\Lambda y
\left({\Pi_1x+(1-x)}\right)]^2}$.
This function is strictly increasing for $y\geq0$. In particular, $\psi(\Pi_1)<\psi(1)$ for
$\Pi_1<1$.
Referring to~\eqref{Dx_expanded}, we have shown $D(x)<0$ for all $0\leq x\leq1$.
\end{proof}

Lemmas~\ref{Lem:positive_externality} and~\ref{Lem:D(x^*)<0} have an important corollary:
$G_{E}(x)$ and $G_{E'}(x)$ each can cross $G_N(x)$ at most once in the interval of $(0,1)$.  Recalling that for $c_0>0$  $G_E(x)<G_{E'}(x)$,  in the most general case, we can introduce $\zeta$ and $\zeta'$ as following:
\begin{multline*}
\exists \zeta\ \& \ \zeta', \ \ 0<\zeta<\zeta'<1, \ \ \text{s.t.}:\\ \begin{cases}
G_{N}(x)<G_E(x)<G_{E'}(x)\ \text{for}\ x\in[0,\zeta)\\
G_{E}(x)<G_N(x)<G_{E'}(x)\ \text{for}\ x\in(\zeta,\zeta')\\
G_E(x)<G_{E'}(x)<G_N(x)\ \text{for}\ x\in(\zeta',1]
\end{cases}
\end{multline*}
The extension to other special cases is straightforward.\footnote{For instance, in the trivial cases
such as when $G_E(x)>G_N(x)$ for all $x\in(0,1)$ (sufficient to check if $G_E(1)>G_N(1)$), then the equilibrium point of the system is the unique and stable point $(y^*,x^*)=(0,1)$. As another example, if $G_{E'}(x)<G_N(x)$ for all $x\in(0,1)$ (sufficient to check if $G_{E'}(0)<G_N(0)$) then the equilibrium point is trivially $(y^*,x^*)=(1-x_0-y_0,0)$.}
Equilibrium points are derived by equating both $\dot{x}$ and $\dot{y}$ in~\eqref{sys_dyn_gen} to zero.
$\dot{y}(t)=0$ yields $y^*=0$ or $v_{01}=0$. Hence:
\begin{gather*}
\dot{x}(t)=0\Leftrightarrow
\begin{cases}  y^*=0\  \&\ \gamma (1-x^*)v_{0'1}-\gamma x^*v_{10'}=0\\
v_{01}=0\ \& \  \gamma (1-x^*-y^*)v_{0'1}-\gamma x^*v_{10'}=0
\end{cases}
\end{gather*}
This leads to:
\begin{gather*}
\!\!\!\begin{cases}
y^*=0, x^*=1,G_{E'}(1)>G_{N}(1)\leftarrow \text{unacceptable}\\
y^*=0,x^*=0,G_{E'}(0)<G_{N}(0)\leftarrow \text{unacceptable}\\
y^*=0,x^*=\zeta'\\
G_{N}(x^*)>G_E(x^*),x^*+y^*=1,G_{E'}(x^*)>G_{N}(x^*)\\
\phantom{wooozzzzaaaaaaaaaaaaaaaa}\Rightarrow x^*\in(\zeta,\zeta'), y^*=1-x^*\\
G_{N}(x^*)>G_E(x^*),G_{E'}(x^*)=G_{N}(x^*)\\
\phantom{wooozzzzaaaaaaaaaaaaaaaa}\Rightarrow x^*=\zeta', y^*\in[0,1-\zeta']\\
{\textstyle G_N(x^*)>G_E(x^*), G_{N}(x^*)>G_{E'}(x^*), x^*=0 \leftarrow \text{unacceptable}}
\end{cases}
\end{gather*}
In brief, the equilibrium point is \underline{not} unique.  First, note that any point $(y^*,x^*)$  for
$x^*\in(\zeta,\zeta')$ and $y^*=1-x^*$  can be an equilibrium point. Also, $(y^*,x^*)$ where $x^*=\zeta'$ and
$0\leq y^*\leq1-\zeta'$ can be equilibrium points.
\hide{To gain an insight about these results, note that when $x^*\in(\zeta,\zeta')$, those ISPs that have purchased the firewall have no incentive to disable them, and those that have not bought it, have no incentive to buy. Specifically, there will be no ISP left that  has purchased the firewall and disabled it.  Also when $x^*=\zeta'$, the ISPs that have purchased the firewall are indifferent between enabling or disabling it, and those that have not purchased it before, have no incentive to buy it then.
}
Which of these equilibrium points is eventually achieved depends on the initial condition that the system starts with.
For instance, if $x(0)=x_0\in (\zeta,\zeta')$, say by seeding the network by free samples, then those ISPs that do not have the firewall have no incentive to obtain it. Hence, in the equilibrium, $y^*=y_0$. 
Now, if $y_0>1-\zeta'$, then the ISPs that have the firewall will enable it until there is no ISP with a disabled firewall, and hence,
$x^*=1-y_0$. On the other hand, if $y_0<1-\zeta'$, then ISPs with firewall will progressively enable it until $x^*=\zeta'$ fraction of the ISPs have enabled their firewalls. Subsequent ISPs with the firewall have no incentive to enable theirs.
A practically interesting case is when the system starts with no seeding, that is $(y_0,x_0)=(1,0)$. In this case $x(t)$ progressively rises to slightly above $\zeta$ and then the system is \emph{locked} into the point $(\zeta^+,1-\zeta^+)$.  For practical purposes, as long as the value of the equilibrium level is concerned, we can  (and henceforth will)  take $(y^*,x^*)=(1-\zeta,\zeta)$ as the equilibrium point for the case of $(y_0,x_0)=(1,0)$.
In summary, we have the following theorem:
\begin{Theorem}
 If firewalls do not cooperate, then an equilibrium point is not unique and depends on the initial value. However, for any given initial value, the equilibrium point is unique and stable.
\end{Theorem}

Fig.~\ref{fig:phase_portrait_1} depicts an example phase portrait for a case in which $c_0>0$.
\begin{figure}
\centering
\includegraphics[width=0.8\linewidth]{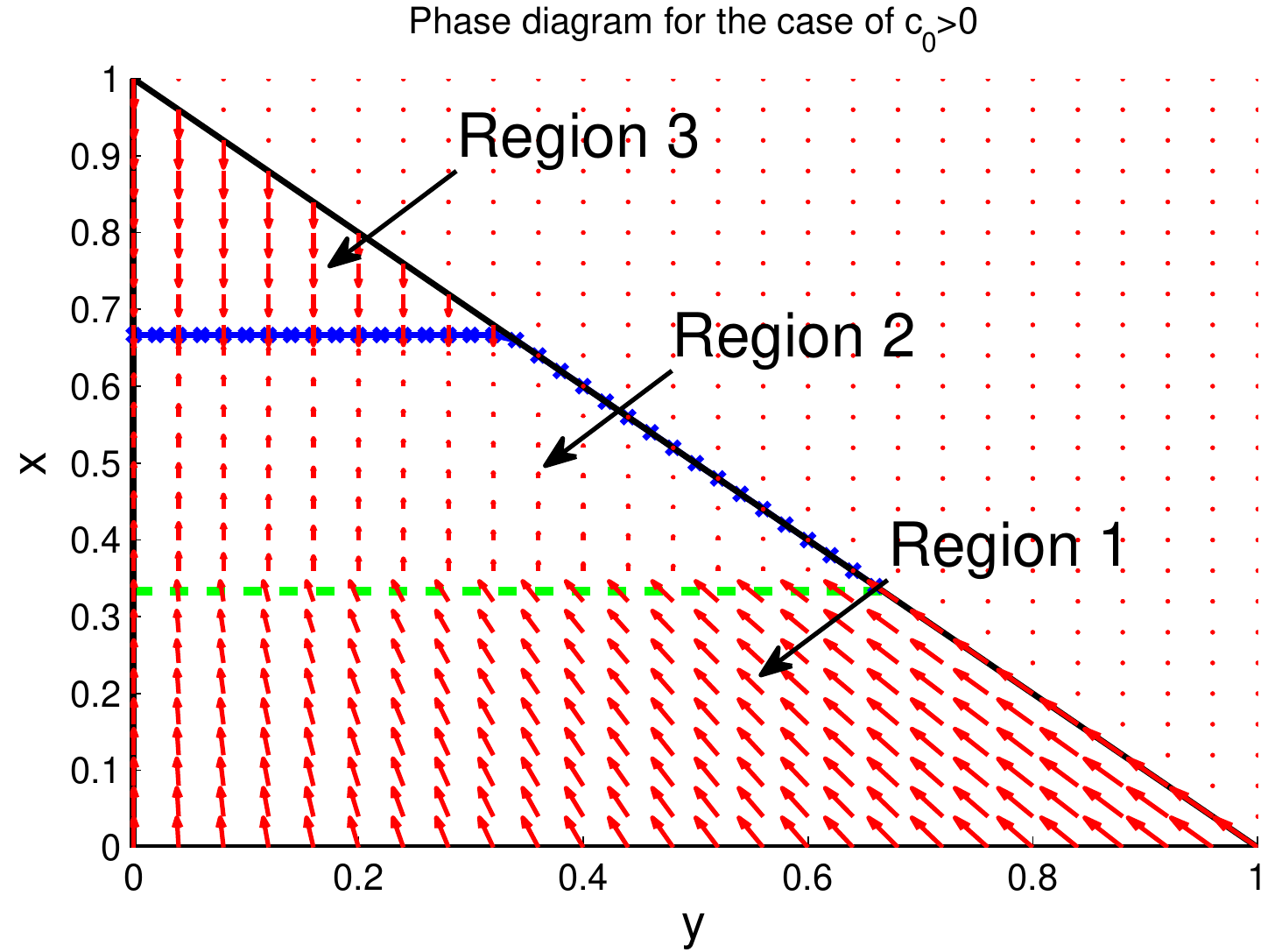}
\caption{A sample phase portrait for the ODE of~\eqref{sys_dyn_gen}. Region 1 is designated by $x<\zeta$, $x+y\leq1$, where the ODE turns to $\dot{y}(t)=-\gamma y(t)$, $\dot{x}(t)= \gamma y(t)+ \gamma (1-x(t)-y(t))$;
Region 2 is where $\zeta<x<\zeta'$, $x+y\leq1$ and the ODE is $\dot{y}(t)=0$, $\dot{x}(t)= \gamma (1-x(t)-y(t))$; Region 3 is the region $\zeta'<x$, $x+y\leq1$, where the ODE is transformed to $\dot{y}(t)=0$,
$\dot{x}(t)= - \gamma x(t)$. All the points in the set of $(\zeta',y)$ for $0\leq y\leq 1-\zeta'$ and $(x,1-x)$ for $\zeta< x\leq \zeta'$ compose the equilibria of the system.}%
\label{fig:phase_portrait_1}%
\end{figure}%
An implication of the theorem is that when firewalls do not cooperate, we do \emph{not} have the phenomenon of \emph{tipping point} (a.k.a. \emph{critical mass}), as they are associated with unstable equilibrium points.
For the rest of the paper, we assume that the system starts from $(y_0,x_0)=(1,0)$.
Moreover, without loss of generality, we assume that
the equilibrium point is non-trivial.
A corollary of this assumption and the assumption of starting from $(y_0,x_0)=(1,0)$ is that at the equilibrium point $(y^*,x^*)$ (which is $(1-\zeta,\zeta)$), we definitely have $G_E(x^*)-G_N(x^*)=0$.

\subsection{The Effect of $\Pi_1$\hide{on $x^*$} and the Phenomenon of \textbf{Free-Riding}}\label{subsec:Pi_1_effect}
Here, we show that \emph{less}
protection on the outgoing traffic by the firewall (hence a larger $\Pi_1$), leads to a \emph{higher} equilibrium level of adoption.
Equivalently, the more firewall provides protection against the outgoing threats, the less ISPs will be willing to adopt the firewall in the equilibrium. Intuitively, with lower ability of the firewall to block the outgoing threats, ISPs feel the need to block the incoming threats instead of free-riding on others' investments.\footnote{There is a cautionary pitfall in this intuitive argument: one could argue that as more ISPs adopt the firewall, its efficacy increases and hence it becomes a more attractive service, leading to higher level of adoption. To counter this argument, the lower bound on $\pi_1$ is necessary. That is, the improvement in the efficacy of the firewall is overwhelmed by the improvement in the efficacy of the firewall for non-adopters.}

\begin{Theorem}\label{Thm:effect_of_omega_1_free_riding}
 When firewalls do not cooperate, $\dfrac{\partial x^*}{\partial \Pi_1} > 0$.
\end{Theorem}
\begin{proof}
From our discussion in~\S\ref{subsec:uniqueness}, for $(y_0,x_0)=(0,1)$, $x^*$ is the solution of
$G_E(x)=G_N(x)$.
Define:
\begin{align}\label{Def:phi}
\phi(p_1,p_0):=\Lambda+\mu(p_1x^*+p_0(1-x^*)).
\end{align}
 Now, we can write $G_E(x)=G_N(x)$ as:
\begin{multline}\frac{c+rc_0}{C_{1I}(\mu+r)+C_{2I}}=\\
\Lambda\left[\frac{\left({\Pi_1x^*+\Pi_0(1-x^*)}\right)}{\phi(\Pi_1,\Pi_0)}-\frac{\left({\pi_1x^*+\pi_0(1-x^*)}\right)}{\phi(\pi_1,\pi_0)}\right]\label{general_nec_condition}
\end{multline}
We can take the partial derivative of both sides of~\eqref{general_nec_condition} with respect to $\Pi_1$:
\begin{align*}
0=\frac{1}{a}D(x^*)\frac{\partial x^*}{\partial \Pi_1}+\frac{\mu x^*}{[\mu+\Lambda\left({
\Pi_1x^*+\Pi_0(1-x^*)}\right)]^2}
\end{align*}
From Lemma~\ref{Lem:D(x^*)<0} in \S\ref{subsec:uniqueness} we have $D(x^*)< 0$. The statement of the theorem follows.
 \end{proof}

\subsection{Pareto-Optimum Adoption and Price of Anarchy}
Here we investigate the scenario in which a social planner can impose the choice of adoption on ISPs, and attempts to maximize the aggregate expected utility. 
The optimization problem of such central planner is the to maximize the \emph{social} utility as defined bellow:
\begin{align}\label{social_util}
 U(x):= xG_E(x)+(1-x)G_N(x)
\end{align}
We denote an optimum $x$ by $\hat x$, i.e.  $\hat x\in\arg\max U(x),0\leq x\leq 1$. If $0<\hat x<1$, it must satisfy the following gradient condition:
\begin{align*}
&G_E(\hat x) -G_N(\hat x) + [x\frac{\partial G_E(x)}{\partial x} +(1-x)\frac{\partial G_N(x)}{\partial x}]\mid_{x=\hat x}=0.
\end{align*}
For $\Pi_1<\Pi_0$, both $\dfrac{\partial G_E(x)}{\partial x}$ and $\dfrac{\partial G_N(x)}{\partial x}$ are strictly positive for all  $0\leq x\leq 1$ according to Lemma~\ref{Lem:positive_externality}.
Therefore, either $\hat x=0$, or $\hat x =1$, or $G_E(\hat x) -G_N(\hat x)<0$.
Recall that $x^*$, the equilibrium point if ISPs could decide themselves, satisfied the condition $G_E(x^*) -G_N(x^*)=0$. Also recall that for $\Pi_1<\Pi_0$, $G_E(x) -G_N(x)$ is a strictly decreasing function of $x$ (by Lemma~\ref{Lem:D(x^*)<0}). Hence we have the following result:
\begin{Theorem}\label{Thm:hat_x}
The socially optimum (planner-imposed) level of adoption, $\hat x$, is always greater than the equilibrium level of adoption, $x^*$.
Moreover, when $0<x^*<1$ and $\Pi_1<\Pi_0$,  then $\hat x>x^*$. 
\end{Theorem}

Theorem~\ref{Thm:hat_x} has important corollaries:
Each ISP enjoys a higher utility in the socially optimum fraction of adoption.  This is because following Lemma~\ref{Lem:positive_externality}, both $G_E(x)$ and $G_N(x)$ are  increasing in $x$ and according to the theorem, $\hat x \geq x^*$.
Then a question may arise as to what prevents the ISPs from reaching this pareto-optimum level of adoption in which everybody is better off. It is because some ISPs are \emph{more} better off than the others. Specifically, those ISPs that do not adopt would enjoy a higher utility than those who adopt it in the socially optimum solution. If the ISPs could freely decide, then they stop adopting once opting out starts to yield more utility, even though continuing to adopt will also increase their utility (but less than opting out would increase).

The price of anarchy ($PoA$) is defined as the ratio of the optimum social utility over the social utility computed for the worst equilibrium.
The price of stability ($PoS$) is defined in a similar way, with the difference that the worst equilibrium is replaced by the best equilibrium. As we showed in~\S\ref{subsec:uniqueness}, for any given initial condition, the equilibrium point in our problem is unique, and hence, the price of anarchy and stability (given an initial value)  are the same:
\[
 PoA=PoS=\frac{\hat x G_E(\hat x)+(1-\hat x)G_N(\hat x)}{x^*G_E(x^*)+(1-x^*)G_N(x^*)}
\]
For $\Pi_1=\Pi_0$, that is, when there is no protection against the outgoing threats, the price of anarchy is unity. 
In the next section, we investigate the ``best'' choice of $\Pi_1$ if it could be imposed by a regulator.

\section{The Problem of a Regulator: Setting the Optimum Protection on Outgoing Traffic}\label{sec:regulator_best_omega_1}
We hitherto assumed that the protection level on the outgoing traffic (and hence $\Pi_1$) is a given parameter of the firewall.
In general, $\Pi_1$ can be any value between $\Pi_0$ (no protection) and $\pi_0$ (the same protection on the outgoing traffic as for the incoming traffic).
Recall that lesser values of $\Pi_1$ correspond to more protection against the outgoing threats.
Theoretically, developers of the firewalls can add protection on the outgoing traffic at least as much as the protection on the incoming traffic at no additional cost of production. 
\hide{This added protection should not be attached to any additional cost of production, as the developed technology is already there for the incoming traffic. Any resulting change in the demand also should not affect the price of the firewall, as the marginal cost of producing more copies of the same software is the same as making one instance of the software.  Hence, under the \emph{perfect competition} assumption for the market, that is the price of the firewall is only determined by its marginal production cost,}
In Theorem~\ref{Thm:effect_of_omega_1_free_riding} in \S\ref{subsec:Pi_1_effect}, we showed that $\dfrac{\partial x^*}{\partial \Pi_1}<0$, therefore, the firewall developers have an incentive to remove any protection against the outgoing threats to maximize their sale.\footnote{It is assumed that the firewall developers make (a small amount of) profit on each copy of the firewall sold.}

The regulator can set a minimum requirement on the amount of protection against the outgoing threats relative to the protection provided against the incoming threats by any firewall. In the language of our model, there can be a maximum value imposed on $\Pi_1$ relative to $\pi_0$. The firewall developers will then choose to produce firewalls with this highest value of  $\Pi_1$, as following Theorem~\ref{Thm:effect_of_omega_1_free_riding}, this will lead to the largest sales for them.
A natural idea seems to put the bound on $\Pi_1$ to be equal to $\pi_0$ as it corresponds to the highest protection against the outgoing as well as the incoming threats. However, as we showed in Theorem~\ref{Thm:effect_of_omega_1_free_riding}, this will lead to free-riding among the ISPs and a low equilibrium level of adoption, hence potentially a less secure network. On the other hand, if the bar on $\Pi_1$ is $\Pi_0$, then the highest level of adoption is achieved but the amount of protection against the outgoing threats is the weakest (nonexistent). Hence, a choice of the ``best'' $\Pi_1$ is a non-trivial question and depends on the metric and the perspective used. In what follows, we consider some of these different metrics.
\subsection{View 1: Decentralized Social Optimum}
Let us define a solution of the following optimization for $\Pi_1$, a \emph{decentralized social optimum} $\Pi_1$:
\begin{align*}
\Pi_1\in\arg\max_{\pi_0\leq \Pi_1\leq \Pi_0} U(x^*)
\end{align*}
where $U(x^*)=x^*G_E(x^*)+(1-x^*)G_N(x^*)$, is the social utility of the ISPs.
We use the phrase \emph{decentralized} to emphasize that ISPs are allowed to freely choose their individual best decision of adoption, and the system is accordingly at the equilibrium value $x^*$.
Assuming $0<x^*<1$, $x^*$ satisfies $G_N(x^*)=G_E(x^*)$. Hence, the above optimization is transformed to:
\[
\Pi_1\in\arg\max_{\pi_0\leq \Pi_1\leq \Pi_0} G_E(x^*),
\]
or equivalently
\begin{align} \label{eq:Pi_1_opt}
 \Pi_1\in\arg\max_{\pi_0\leq \Pi_1\leq \Pi_0} G_E(x) \ \ s.t. \ \ G_E(x)-G_N(x)=0
\end{align}
The above optimization is non-convex (due to the nonlinear equality constraint). Nevertheless, we can observe the following theorem without solving the optimization.
\begin{Theorem}\label{Thm:socialy_optimum_Pi_1}
In non-cooperative firewalls, for the case of $\pi_1=\pi_0$ (i.e., mutually inclusive firewalls), any $\Pi_1$ is a decentralized social optimum.
\end{Theorem}

In words, if the ISPs are allowed to individually take their adoption decisions, then the optimum value of the social utility is not affected by the level of protection imposed on the outgoing traffic.
Now, initially this might come as a surprise; after all, the value of $x^*$ \emph{does} change with changing $\Pi_1$. The proof of the theorem reveals why this does not affect the value of social utility.
Intuitively, it is because as $\Pi_1$ is decreased, the protection provided by the firewall improves, however, a smaller fraction of the ISPs end up adopting the firewall. The aggregate impact is that the overall utility in the network stays unaffected.

\begin{proof}
In~\eqref{eq:Pi_1_opt}, for $\pi_0=\pi_1$, $G_E(x)$ does not depend on either $x$ or $\Pi_1$.
\end{proof}

The proof may suggest that the above observation may only hold for the case of $\pi_1=\pi_1$, i.e. mutually exclusive scenario. Recall that this is case in which the blocking outcome of the firewalls of the intruder and of the victim are the same, that if if an intrusion can bypass one, it can bypass the other as well. 
Further investigation, however, reveals that in fact, this result can be generalized to other scenarios too:
\begin{Theorem}\label{Thm:socialy_optimum_Pi_1_Generalised}
In non-cooperative firewalls, for any case of $\pi_1=\pi_0\Pi_1+\alpha(\pi_0-\pi_0\Pi_1)$ for an $\alpha\in[0,1]$, any $\Pi_1$ is decentralized-socially optimum.
\end{Theorem}

Note that for $\alpha=0$ we have the case in which the blocking outcomes of the two firewalls are independent of each other, and for $\alpha=1$, the mutual inclusive case as in Theorem~\ref{Thm:socialy_optimum_Pi_1}. \hide{The proof is provided in our technical report \cite{CDC2012_tech_rep}.}

\begin{proof}
We need to show $\dfrac{dG_1(x^*)}{d\Pi_1}=0$. A change in $\Pi_1$ affects the utility at the equilibrium both directly (dependence of $G_N,G_E$ on $\Pi_1$) and indirectly through changing the equilibrium level $x^*$. Hence:
\begin{align}\label{eq:d_G_1_d_Pi_1}
\frac{dG_1(x^*)}{d\Pi_1}=\frac{\partial G_1(x^*)}{\partial \Pi_1}+\frac{\partial G_1(x^*)}{\partial x^*}\frac{\partial x^*}{\partial \Pi_1}
\end{align}
We have:
\begin{align*}
\frac{\partial G_1(x^*)}{\partial \Pi_1}=&\mu x^*\left(\Lambda+\mu\left(\Pi_1x^*+\Pi_0(1-x^*)\right)\right)^{-2}\\
\frac{\partial G_0(x^*)}{\partial \Pi_1}=&(\mu x^*\frac{\partial \pi_1}{\partial \Pi_1})\left(\Lambda+\mu\left(\pi_1x^*+\pi_0(1-x^*)\right)\right)^{-2}\\
\frac{\partial G_1(x^*)}{\partial x^*}=&-(\Pi_0-\Pi_1)x^*\left(\Lambda+\mu\left(\Pi_1x^*+\Pi_0(1-x^*)\right)\right)^{-2}\\
\frac{\partial G_0(x^*)}{\partial x^*}=&-(\pi_0-\pi_1)x^*\left(\Lambda+\mu\left(\pi_1x^*+\pi_0(1-x^*)\right)\right)^{-2}
\end{align*}
To calculate $\dfrac{\partial x^*}{\partial \Pi_1}$, we note that from $G_1(x^*)=G_0(x^*)$ we can deduce:
\begin{gather*}
 \frac{\partial G_1(x^*)}{\partial \Pi_1}+ \frac{\partial G_1(x^*)}{\partial x^*}\times\frac{\partial x^*}{\partial \Pi_1}= \frac{\partial G_0(x^*)}{\partial \Pi_1}+ \frac{\partial G_0(x^*)}{\partial x^*}\times\frac{\partial x^*}{\partial \Pi_1}\\
\Rightarrow \frac{\partial x^*}{\partial \Pi_1}=-\frac{\frac{\partial G_1(x^*)}{\partial \Pi_1}-\frac{\partial G_0(x^*)}{\partial \Pi_1}}{\frac{\partial G_1(x^*)}{\partial x^*}-\frac{\partial G_0(x^*)}{\partial x^*}}
\end{gather*}
Referring to the definition of $\phi(p_1,p_0)$ in \eqref{Def:phi}, and replacing the terms in~\eqref{eq:d_G_1_d_Pi_1}, we obtain:
\begin{multline*}
\frac{dG_1(x^*)}{d\Pi_1}=\phi(\Pi_1,\Pi_0)^{-2}\times\\
\Bigg[\mu x^* - \mu(\Pi_0-\Pi_1)\frac{x^*}{(\Pi_0-\Pi_1)}\times\\
 \frac{\mu\phi(\Pi_1,\Pi_0)^{-2}
-\mu\frac{\partial \pi_1}{\partial\Pi_1}\times\phi(\pi_1,\pi_0)^{-2}}
{\mu\phi(\Pi_1,\Pi_0)^{-2}
-\mu\frac{\pi_0-\pi_1}{\Pi_0-\Pi_1}\phi(\pi_1,\pi_0)^{-2}}\Bigg]=0
\end{multline*}

The last equality follows because for $\pi_1=\pi_0\Pi_1+\alpha(\pi_0-\pi_0\Pi_1)$, we have
$\dfrac{\partial \pi_1}{\partial\Pi_1}=\dfrac{\pi_0-\pi_1}{\Pi_0-\Pi_1}$.
\end{proof}
Theorems~\ref{Thm:socialy_optimum_Pi_1}, \ref{Thm:socialy_optimum_Pi_1_Generalised} establish that if ISPs freely take their adoption decisions, then the social optimum utility is not affected by the level of protection on the outgoing traffic. Hence, it is interesting to investigate other metrics of optimality for $\Pi_1$, as we do next.
\subsection{View 2: Decentralized Individual Optimum}
For individuals adopting the firewall, a decentralized optimum $\Pi_1$ is given by: $\Pi_1\in\arg\max_{\pi_0\leq \Pi_1\leq \Pi_0} G_E(x^*).$
For individuals who opt out, an optimum $\Pi_1$ is provided by:
$\Pi_1\in\arg\max_{\pi_0\leq \Pi_1\leq \Pi_0} G_N(x^*).$
Since at equilibrium, we have $G_N(x^*)=G_E(x^*)$, solutions of the above two optimizations are the same; and are further equal to the solution of the social optimum utility as was shown in~\eqref{eq:Pi_1_opt}.
Hence, a decentralized social optimum value of $\Pi_1$, is also decentralized individual optimum for all ISPs.
Moreover, Theorem~\ref{Thm:socialy_optimum_Pi_1_Generalised} applies to the decentralized individual optimum $\Pi_1$ as well.

\subsection{View 3: Decentralized Security Optimum}
Since, both decentralized social and decentralized individual utilities are unaffected by the value of $\Pi_1$,
we define a new metric of practical interest. One can consider only the cost of the intrusions in the network to be the metric of optimality.
Accordingly, we define the \emph{security} utility $V(x)$ to be the negative of the expected aggregate damage incurred on the network as a result of the intrusions if the adoption decisions are taken in a decentralized manner by the ISPs. Note specifically that $V(x)$ does not include the costs of the adoption of the firewall.
In our problem, $V(x)$ is as follows:
\begin{align}\label{security_utility}
 V(x)=x\left(G_E(x)+\frac{c}{r}+c_0\right)+(1-x)G_N(x)
\end{align}
It is easy to see that $V(x)$ is increasing in $x$. Hence, in the light of Theorem~\ref{Thm:hat_x}, $V(\hat x)\geq V(x^*)$. That is, the socially optimum level of adoption not only provides a better aggregate expected utility (by construction), it also provides better overall network security, compared to the equilibrium level of adoption.
If we were to \emph{centrally} maximize $V(x)$, then the optimum choice for $x$ (denoted by $\tilde x$) would be $\tilde x=1$. Note that the adopters in this case win a still higher utility than the adopters even in the socially optimum scenario. However, this is achieved at the cost of lower utility realization for the non-adopters in the socially optimum choice of $x$ (hence its achieved aggregate expected utility is lower).

Now, as we did in the previous two subsections, we define a \emph{decentralized security optimum} $\Pi_1$ to be a  solution of the following optimization problem: $\max_{\pi_0\leq\Pi_1\leq \Pi_0} V(x^*)$.
The socially optimum $\Pi_1$ is provided by the following theorem:
\begin{Theorem}\label{Thm:security_optimum_Pi_1}
In non-cooperative firewalls, $\Pi_1=\Pi_0$ is decentralized security optimum choice for $\Pi_1$.
\end{Theorem}

Recall that $\Pi_1=\Pi_0(=1)$ translates to the provision of no protection on the outgoing traffic.

\begin{proof}
From Theorem~\ref{Thm:socialy_optimum_Pi_1_Generalised}, changing $\Pi_1$ does not affect the value of $U(x^*)$.
From the definition of $V(x)$ in~\eqref{security_utility}, we have: $V(x^*)=U(x^*)+x^*(c/r+c_0)$.
Hence, a $\Pi_1$ that maximizes $V(x^*)$ must maximize $x^*(c/r+c_0)$. For any non-trivial firewall, we have $c/r+c_0>0$.
From Theorem~\ref{Thm:effect_of_omega_1_free_riding}, higher $x^*$ is achieved by increasing $\Pi_1$. Therefore, $\Pi_1=\Pi_0$ maximizes
$x^*$, and hence $V(x^*)$.
\end{proof}
\hide{A point to recall is that for $\Pi_1=\Pi_0$, the price of anarchy is unity. That is, for $\Pi_1=\Pi_0$ fixed, the equilibrium level of adopters maximizes the expected aggregate social utility as well.
However,} 
Theorem~\ref{Thm:security_optimum_Pi_1} is in fact a negative result: the highest security utility in the decentralized case is achieved if the regulator requires no protection on the outgoing traffic. Next, we discuss how this reveals a significant inefficiency.

\subsection{Centralized Optimum}
We can consider the optimum centralized version of each of the previous three viewpoints: social, individual and security.
For a centralized optimizer, the choice of variables $x^*$ and $\Pi_1$ are \emph{decoupled}. For any $\Pi_1<\Pi_0$, we have $\dfrac{\partial G_N(x)}{\partial x},\dfrac{\partial G_E(x)}{\partial x}>0$. Also, for any $x$, we have $\dfrac{\partial G_N(x)}{\partial \Pi_1}<0$, $\dfrac{\partial G_E(x)}{\partial \Pi_1}\leq 0$. Hence, the centralized optimum value of $\Pi_1$ that simultaneously maximizes $U(x)$, $V(x)$,
$G_E(x)$, $G_N(x)$ is $\Pi_{1,\min}$, i.e., maximum protection on the outgoing traffic.
Note that this is despite potentially different centrally optimum $x$ for each of these optimizations.
For instance, the pair $(\Pi_1,x)$ that maximizes the social utility is $(\Pi_{1,\min},\hat{x})$,  while the pair $(\Pi_1,x)$ that maximizes the security utility is $(\Pi_{1,\min},1)$, and it is possible to have $\hat x<1$.

In the decentralized scenario, i.e. when ISPs are free to adopt or not, maximum security utility is achieved \emph{at the cost of} eliminating the protection against the outgoing threats, as the values of $x^*$ and $\Pi_1$ were coupled. This suggests that a regulation on only firewalls is inefficient, and an efficient regulation should be on both firewalls and ISPs.

\section{The Price of Shortsightedness}
\subsection{The Effect of the Discount Factor on the Equilibrium}\label{subsec:the_most_general_effect_of_r}
We showed in \S\ref{subsec:uniqueness} that when the market is not seeded, i.e. $(y_0,x_0)=(1,0)$, the equilibrium fraction of the nodes that adopt and enable the firewall is (practically) equal to $\zeta$, which is the solution of $G_{E}(x)=G_{N}(x)$, noting that $G_{E}(x)$ includes negative of the purchase fee, $-c_0$, as well. Recall the definition of $\phi$ as in~\eqref{Def:phi}, that is $\phi(p_1,p_2):=\mu+\Lambda\left(p_1x+p_0(1-x)\right)$. Taking the partial derivative of both sides of the equation $G_E(x)=G_N(x)$ w.r.t. $r$ leads to:
\begin{multline*}
 \!\!\!\!\!\!\!\!\!\hide{\frac{c+rc_0}{C_{1I}(\mu+r)+C_{2I}}=\\
\Lambda\left[\frac{\left({\Pi_1x^*+\Pi_0(1-x^*)}\right)}{\phi(\Pi_1,\Pi_0)}-\frac{\left({\pi_1x^*+\pi_0(1-x^*)}\right)}{\phi(\pi_1,\pi_0)}\right]\Rightarrow\label{general_nec_condition}\\}
 \frac{c_0(C_{1I}\mu+C_{2I})-cC_{1I}}{[C_{1I}(\mu+r)+C_{2I}]^2}=
\underbrace{\mu\Lambda\left[\frac{\left({\pi_0-\pi_1}\right)}
{\phi(\pi_1,\pi_0)^2}-\frac{\left({\Pi_0-\Pi_1}\right)}
{\phi(\Pi_1,\Pi_0)^2}\right]}_{E(x^*)}\frac{\partial x^*}{\partial r}
\end{multline*}
In order to determine the sign of $\dfrac{\partial x^*}{\partial r}$, we investigate the sign of its coefficient, denoted by $E(x^*)$. From~\eqref{Dx_expanded} it follows that $E(x^*)=D(x^*)/a$.
From Lemma~\ref{Lem:D(x^*)<0}, $D(x^*)<0$.
Hence, we have the following result:
\begin{Theorem}\label{Thm:general_effect_of_r}
$\sgn{\dfrac{\partial x^*}{\partial r}} =\sgn(cC_{1I}-c_0(C_{1I}\mu+C_{2I})).$
\end{Theorem}

\hide{When $c_0>0$, the effect of the discount factor on the equilibrium level of adoption, non-trivially depends on the relative weights of the different components of the costs.}
Consider the extreme case of $\mu=0$, that is, a successful intrusion is never detected. According to the theorem, in this case, if $\dfrac{c_0}{c}<\dfrac{C_{1I}}{C_{2I}}$, then $\dfrac{\partial x^*}{\partial r}>0$. That is, if the ratio of the purchase fee to the (per unit time) usage cost of the firewall is less than the ratio of instantaneous cost to per unit cost of the infection, then \emph{more shortsightedness} leads to \emph{higher} equilibrium levels of adoption. Intuitively, it is because ISPs tend to \emph{see} the more immediate effects more clearly (i.e., weigh them more). Hence in this case, where the firewall carries a lower weight of instantaneous cost relative to per unit cost than the infection does, the firewall becomes a more attractive service as ISPs become more shortsighted.
\subsection{Measuring the Price of Shortsightedness}
We introduce a measure of inefficiency as a result of the ISPs being shortsighted in their decision taking, that is, the effect of discounting future events and hence having a bias toward near future outcomes.
To make this formal, we first discuss the notion of a \emph{reference} discount factor $r_0$.
The expected utility associated to an achieved equilibrium can be measured using a potentially different discount factor from the one used by ISPs in their calculation of individual utilities. We have used $r$ to denote the discount factor used by ISPs. We will use $x^*(r)$ to represent the achieved equilibrium level of adoption to emphasize the dependence of $x^*$ on the discount factor of the ISPs.
Let $r_0$ be the discount factor used by a referee. For instance, $G_E(x^*(r),r_0)$ is the expected $r_0$-discounted utility of adopters in an equilibrium level $x^*$ where the ISPs' discount factor is $r$.

We can now formally define the Price of (temporal) Shortsightedness.
The \underline{So}cial \underline{P}rice \underline{o}f \underline{Sh}ortsightedness ($SoPoSh$) is defined as follows:
\[
SoPoSh(r):=\lim_{r_0\to 0}\frac{U(\lim_{r\to 0}x^*(r),r_0)}{U(x^*(r),r_0)}
\]
Note the dependence of $SPoSh$ on $r$. Also, since we showed $x^*(r)$ is unique, there is no ambiguity between the choices of equilibria. 
Similarly, we can define the \underline{Se}curity \underline{P}rice \underline{o}f \underline{Sh}ortsightedness $SePoSh(r)$ as $\frac{V(x^*(r\to0),r_0\to0)}{V(x^*(r),r_0\to0)}$.

Following its definition, higher prices of shortsightedness means that from the referee's viewpoint, shortsightedness of the decision-takers leads to inefficiency. As $r$ approaches 0, these prices of shortsightedness approaches one. However, unlike Price of Anarchy and Stability, there is no general rule that the prices of shortsightedness is always less (or greater) than unity.
That is, depending on the parameters of the problem, it might be ``beneficial'' to have less or more shortsighted decision-takers. This has interesting policy-making implications. In general, we have the following theorem:
\begin{Theorem}\label{Thm:PoSh}
For all $r$, both $SoPosh(r)$ and $SePoSh(r)$ are: (a)~greater than one if $cC_{1I}\geq c_0(C_{1I}\mu+C_{2I})$;
(b)~less than one if $cC_{1I}\leq c_0(C_{1I}\mu+C_{2I})$;
(c)~equal to one if $cC_{1I}=c_0(C_{1I}\mu+C_{2I})$.
\end{Theorem}
\begin{proof}
We present the proof for $SoPoSh(r)$; the proof for $SePoSh(r)$ follows similarly.
If $\dfrac{\partial U(x^*(r),r_0)}{\partial r}<0$, then $SoPoSh(r)>1$, and so on.
We have: $\dfrac{\partial U(x^*(r),r_0)}{\partial r}=\dfrac{\partial U(x^*(r),r_0)}{\partial x^*(r)}\times\dfrac{\partial x^*(r)}{\partial r}$.
$U(x^*(r),r_0)=x^*(r)G_E(x^*(r),r_0)+(1-x^*(r))G_N(x^*(r),r_0)$, hence: $\dfrac{\partial U(x^*(r),r_0)}{\partial x^*(r)}=[G_E(x^*(r),r_0)-G_N(x^*(r),r_0)]+x^*(r)\dfrac{\partial G_E(x^*(r),r_0)}{\partial x^*(r)}+(1-x^*(r))\dfrac{\partial G_N(x^*(r),r_0)}{\partial x^*(r)}$.
The first term is positive as a consequence of Lemma~\ref{Lem:D(x^*)<0}. Note that in the proof of Lemma~\ref{Lem:D(x^*)<0}, the specific value of $r$ was irrelevant. The second and third terms are also positive (at least one of them strictly positive), as for any $r_0$, we have $\dfrac{\partial G_N(x,r_0)}{\partial x},\dfrac{\partial G_E(x,r_0)}{\partial x}\geq 0$. The theorem now follows from Theorem~\ref{Thm:general_effect_of_r}.
\end{proof}

\hide{

\section{Alternative Cost Model}
In \S\ref{Sec:Model}

}

\hide{
\section{Miscellaneous Discussions}
\subsection{Critical Mass Phenomenon?}
\subsection{Positive Internality of $\Pi_1$}\label{subsec:internality_of_Omega_1}
}

\section{Conclusion}

We present a new analytical model for the adoption of security measures by autonomous systems, particularly in the context of asymmetric, bidirectional firewalls.  Our analysis provides insights into equilibrium and stability of adoption levels as well as policy perspectives on socially optimal outcomes, price of anarchy, and price of shortsightedness.  We highlight several interesting results, including the dependence of the equilibrium adoption on the initial seeding,  how performance improvements in blocking outgoing threats can decrease the overall firewall adoption due to free-riding, that socially optimal solution increases \emph{every} ISP's utility, and that shortsightedness of ISPs may sometimes help in the overall network to become more secure.  In the future, we intend to extend this analytical model to explore the effects of heterogeneities in the ISPs (their shortsightedness, costs and subnet sizes) on the adoption process and seeding strategies.


\bibliographystyle{IEEEtran}
\bibliography{IEEEabrv,Adoption_Reference}

 \end{document}